\documentclass[aip,jmp,reprint]{revtex4-1}
\usepackage{bm}
\usepackage{amsthm}
\usepackage{amsmath}
\usepackage{hyperref}

\def\bonlinecite#1{[\onlinecite{#1}]}

\newtheorem{theorem}{Theorem}

\newtheorem{corollary}{Corollary}
\newtheorem{conv}{Convention}
\newtheorem{prop}{Proposition}
\newtheorem{rmk}{Remark}
\newtheorem{lemma}{Lemma}
\newtheorem{claim}{Claim}

\newtheorem{innercustomthm}{Boole's Lemma\!}
\newenvironment{customthm}[0]
{\innercustomthm}
  {\endinnercustomthm}

\newcommand{\inner}[2]{\ensuremath{\left\langle{#1,#2}\right\rangle}}

\begin{document}

\title{Differentiability of correlations in Realistic Quantum Mechanics}

\author{Alejandro Cabrera} \affiliation{Instituto de Matem\'atica,  
 UFRJ, CEP 21941-909, Rio de Janeiro - Rio de Janeiro, Brazil} 
\author{Edson de Faria} \affiliation{Instituto de Matem\'atica e Estat\'\i stica, 
 USP, Rua do Mat\~ao 1010, SP 05508-090 , S\~ao Paulo, S\~ao Paulo, Brazil} 
 \author{Enrique Pujals} \affiliation{IMPA, Estrada Dona Castorina 110, 22460-320, Rio de Janeiro, Rio de Janeiro, Brazil}
 \author{Charles Tresser} \affiliation{IBM, P.O.  Box 218, Yorktown Heights, NY 10598, U.S.A.}

% \email[]{acabrera@labma.ufrj.br}
% \email[]{edson@ime.usp.br}
% \email[]{enrique@impa.br}
% \email[]{charlestresser@yahoo.com}
\date{6 September 2015}

\begin{abstract}
We prove a version of Bell's Theorem in which the Locality assumption is weakened. 
We start by assuming 
%new
theoretical
quantum mechanics and weak forms of relativistic causality and of realism (essentially the fact that observable values are well defined independently 
of whether or not they are measured).  Under these hypotheses, we show that only one of the correlation functions that can be formulated in the 
framework of the usual Bell theorem is unknown.  We prove that this unknown function must be differentiable at certain angular configuration 
points that include the origin. We also prove that, if this correlation is assumed to be twice differentiable at the origin, then we arrive 
at a version of Bell's theorem. On the one hand, we are showing that any realistic theory of quantum mechanics which incorporates the 
kinematic aspects of relativity must lead to this type of \emph{rough} correlation function that is once but not twice differentiable. 
On the other hand, this study brings us a single degree of differentiability away from a relativistic von Neumann no hidden variables theorem.  
\end{abstract}
%%%pacs number: 03.65.Ta
%\tableofcontents
\maketitle

\section{Introduction}

One interpretation of Bell's Theorem \bonlinecite{Bell}, \bonlinecite{Mermin}  establishes that quantum mechanics (QM) is incompatible with the conjunction of \emph{locality} (the absence of any causal relation between events that are space-like separated in the usual sense within the kinematics of special relativity (SR)) and some form of {\emph{realism}} (according to which, colloquially speaking, values of {observables} are well defined independently of their measurement; see a more detailed discussion below).  We refer the reader to the 2014 volume of Journal of Physics A devoted to ``50 years of Bell's theorem" \bonlinecite{BellVolume} and references therein for recent disputes on the nature of the content and on implications of Bell's theorem.

A related 
%new
theoretical %
question that goes back to the early days of modern quantum mechanics reads:
 \emph{Is there a realistic theory compatible with what we know of microscopic physics?}  In a language that was important at the time of \bonlinecite{Neumann}, about 6 years after the formulation of QM, the question reads: \emph{Is there a predictive hidden variables theory with averages predicted by QM (and fully consistent with all that is known of microphysics \bonlinecite{Bohm0})?}  A negative answer would close an old wound: the proof of von Neumann' 1932 theorem \bonlinecite{Neumann} on the non-existence of hidden variables is faulty \bonlinecite{Hermann}, \bonlinecite{Bell2}. The mistake in \bonlinecite{Neumann} has so far not been fully repaired.  Of course, realism is ruled out by {the} Bell-Kochen-Specker theorems \bonlinecite{KS}, now even proven in forms \bonlinecite{Klyachko} that permit experimental verification \bonlinecite{Lapkiewicz}, but this is under an assumption, \emph{non-contextuality}, that is stronger (or rather less compelling) than locality as already discussed by Bell \bonlinecite{Bell2} (\emph{cf.} \bonlinecite{Mermin}).

In this paper we assume \emph{classical microscopic realism} (CMR, see Section \ref{sec:setting}) and look for extra hypotheses as weak as we can find to get  a 
%new
theoretical %
contradiction as in usual Bell theory.  The assumptions made in this paper, to be discussed below, are proven to be weaker than locality and also to allow for possible contextuality.   
  
More precisely, we introduce the Restricted Effect-After-Cause Principle (REACP) as a substitute for the locality hypothesis. Then, assuming QM, CMR and the REACP, we prove, as our Theorem 1, that some correlations that are central in Bell's theory are once differentiable at the origin. Under the further smoothness assumption of twice differentiability of the correlations at the origin, we get, as our Theorem 2, a contradiction as in usual forms of Bell's Theorem.  
The outline and details of the proofs of our main theorems are discussed in the next Section \ref{sec:out}. Before going into those details, we discuss here the nature of our hypotheses.

As a first approximation to the REACP the EACP part of it tells us that \emph{an observable value that is registered, or that would be defined if observable values would pre-exist to their measurement, cannot change because of later causes}. With this in mind, the ``R" 
in front of ``REACP" means that, contrary to the EACP which applies to all observables, 
 \emph{the REACP is only claimed to affect those values that are actually measured}. 
One can already notice that the REACP is at most as strong an hypothesis as the EACP.   Furthermore, we shall show that the REACP is \emph{weaker than locality} and that \emph{it can be deduced from the kinematics of SR} (which is, anyway, already part of the traditional Bell's theorem discussion). 

It follows from Theorems 1 and 2 that {at least one} of QM, SR, CMR and the extra differentiability assumption that we make (on top of the proven first degree of differentiability) has to be false.  At this point there are two possible viewpoints:

\noindent
- On the one hand, one can reasonably assume that said correlations are in fact twice differentiable. This assumption is generally considered as mild in physics when, like in our case, once differentiability has been established. (An exception is provided by second order phase transition, but nothing lets us presume that we are in such a situation.) Moreover, we are, {after all}, in a context where all correlations appearing in any theory {are} either real-analytic or fail to be differentiable. From this viewpoint, since we prove once differentiability while twice differentiability yields the contraction that we seek, we can say that we are one degree of differentiability away from a relativistic von Neumann no-hidden-variables theorem. 

\noindent
- On the other hand, one can instead argue that such differentiable but non-smooth correlations cannot be considered as particularly unnatural when they involve possible but not actual measurements, as it happens in our context (see the outline in Section \ref{sec:out} below).  

At any rate, our theorems show that some ``rough" correlations 
(\emph{i.e.,} correlations that are only mildly smooth, so in particular not analytic) must be part of any realistic theory of the type that we consider.  In particular, notice that Bohm-de Broglie theory, the most developed and most widely accepted realistic hidden variable theory, is claimed by some of its adopters to retain 
SR at the observational level. With such views on SR, this theory is then 
included in our analysis (see Section \ref{sec:last} for a further discussion).

\bigskip

%new
{\bf About the nature of the results}: Since the interpretation and scope of theoretical explorations such as the present ones can be very subtle, we want to include the following clarifications. %

First, we warn the reader that some of the functions relevant to the paper only make sense when assuming CMR: see the statement after Convention 2.
Also notice that all along this text we liberally speak of functions and their smoothness despite the fact that we do not assume non-contextuality (see, \emph{e.g.,} \bonlinecite{Mermin}).  Even if the correlations are not defined at once for more than a finite number of their arguments, one can collect the values in all counterfactual realizations to build the correlation functions that we study.

%new
 Finally, we would like to stress that our arguments work at the \emph{pure theoretical level}, namely, we follow the usual method of testing if a theory satisfying certain hypotheses can exist by means of Gedankenexperiments. In particular, we only consider idealized quantum mechanical settings (e.g. no noise) and,  a priori, no claims about experimental verification are implied. %

\section{Outline of the results}
\label{sec:out} Our chain of arguments goes as follows.
 Let us first recall the well known fact that \emph{any four sequences of elements in $\{-1,1\}$ must satisfy some inequalities discovered by Boole  
 (in work whose relevance to Bell's theory was recognized and first analyzed by Pitowsky \bonlinecite{Pitowsky2001})}.
 Bell's theorem, as it can be read off Bell's 1964 paper \bonlinecite{Bell}, 
 can then be re-stated
 as follows:  \emph{Sequences of spin measurements that coexist if one assumes CMR cannot coexist when one also assumes locality since, in such case, the corresponding Boole inequalities do not hold.} 
 
\noindent
\begin{enumerate}
 \item As in {the} usual derivations of Bell's theorem, we consider sequences of EPR pairs: for each pair the two particles fly apart so that their detections are space-like separated events. This is detailed in section \ref{sec:setting}. The hypothesis of realism (more precisely, CMR) is defined in section \ref{sec:setting}.  The CMR hypothesis  allows us to consider counterfactual measurements, \emph{i.e.,} measurements that could have been made (instead of the factual ones) using alternative orientations of the detectors, thus defining functions of the orientation angles.
 
\item Using effect-after-cause arguments ({stemming} from our REACP and defined in section \ref{sec:REACP}) we show that the value of the measurement actually made by Bob does not depend on whether Alice's measurement of the other particle is actual or counterfactual for a Lorentz observer who sees Bob measurements as occurring 
before Alice's.  This implies that three out of the four correlation functions appearing in the usual Bell-type inequalities can be computed explicitly.
\item At the same time, we recall that the four correlation functions must obey Boole's lemma (section \ref{subsec:boole}, equations \eqref{CHSHeq}) since, by definition, they correspond to correlations of sequences of $\pm 1$.
 \item We prove (Theorem \ref{diff0}) that the fourth correlation, which is an unknown function of three angle differences, must be differentiable on the three coordinate axes and, in particular, at the origin where the differential is zero. 
 \item We further show (Theorem \ref{maint}) that if this fourth correlation is twice differentiable at the origin, then we arrive at a contradiction. These two theorems together show that there is a strong constraint on the type of correlation functions that appear in any realistic theory of QM which is compatible with the REACP.
 \item Finally,  we discuss the REACP (section \ref{sec:disc}) and the overall conclusions (section \ref{sec:last}). We show that the REACP is implied by the kinematic part of SR (Claim \ref{REACPWeaker}) and that it is strictly weaker than the locality hypothesis (Corollary \ref{REACPandRegWeaker}).
\end{enumerate}

\section{Setting and conventions}\label{sec:setting}
We consider sequences $(p_i,\overline{p}_i)$ of \emph{EPRB pairs} of spin-$\frac{1}{2}$ particles with the \emph{singlet state} as wave function's spin part:  
\[
\Psi_i=\frac{1}{\sqrt{2}}\left(| +\rangle _{p_i}\otimes|
-\rangle_{{\overline{p}}_i}-| -\rangle_{p_i}\otimes|+\rangle_{{\overline{p}}_i}
\right)\,. 
\]
The ``EPRB" acronym evokes a reformulation by Bohm \bonlinecite{Bohm} of the \emph{``EPR paper"} \bonlinecite{EPR} using spin-$\frac{1}{2}$ particle pairs.  In what follows, \emph{corresponding measurements} made by Alice and Bob are assumed to \emph{space-like separated} (\emph{i.e.,} $\|\Delta x\|>c\cdot\|\Delta t\|$ for the space and time distances, where $c$ is the speed of light in the vacuum).  These measurements are respectively made by these two agents on the successive members $p_i$ and $\overline{p}_i$ of the successive pairs.  By definition $p_i$ is the member of the $i^{\rm th}$ pair that flies to Alice while $\overline{p}_i$ is the member of that pair that flies to Bob.  Thus Alice measures the normalized projections $a=\{a_i\}$ of the spins of the $(p_i)$'s along the vector ${\bm{v}}_{\theta_a}$ at angle $\theta_a$ while Bob measures the normalized projections $b=\{b_i\}$ of the spins of the $(\overline p_i)$'s along ${\bm{v}}_{\theta_{\bm b}}$.  All our vectors belong to some fixed plane, and the angle $\theta _{\bm c}$ is measured relative to a reference vector ${\bm{v}}_{\bm 0}$ at angle $0$ chosen once and for all in that plane.  We shall set $\theta_{\bm{cd}}\equiv \theta_{\bm c} -\theta_{\bm d}$.  QM teaches us that the correlation $\inner{a}{b}$ is given by the \emph{twisted Malus law}:
\[
\inner{a}{b}\equiv \lim_{n\to\infty}\frac{1}{n} \sum_{i=1}^n
a_i b_i=-\cos(\theta _{\bm{ab}}). 
\]

A  \emph{Bell type experiment} (\emph{i.e,} an experiment of the type considered in Bell's theory: see \bonlinecite{Bell} and subsequent papers on the same theme) concerns two sequences of observed normalized spin projections: $a=\{a_i\}$ and $b=\{b_i\}$.  Predictive Hidden Variables theories have often been proposed as a mean to obtain a deterministic theory for micro-physics.  The characteristic property of these theories of relevance for us here is the fact that they are realistic, so that observable values have a meaning before (and in fact independently of) being measured.  Instead of using such theories, as Bell in \bonlinecite{Bell}, one tends now (see \emph{e.g.,}  \bonlinecite{Stapp1971} and \bonlinecite{Leggett2008}) to rather use one of the weaker realism assumptions that go collectively under the acronym of CMR.  CMR tells us that \emph{all the observables that could be measured have well defined values}.  The weakest form of CMR needed to discuss Bell's Theory is used in \bonlinecite{Tresser1} and \bonlinecite{Tresser2}.  We will stick to CMR to which we have adapted the following two conventions that go back at least to \bonlinecite{Bell}. 
\begin{conv}
{We assume that any quantities that are not measured, but that exist according to CMR and are co-measurable according to QM, have the values that would have been obtained if they had been jointly measured quantities.}
\end{conv}
It is important to note that Convention $1$ allows for possible contextuality, as opposed to the hypothesis of the Kochen-Specker theorem.

\begin{conv}
{We assume that  the statistical predictions (\emph{e.g.,} average values and in particular correlations) 
are the ones given by  QM.}
\end{conv}
From the CMR hypothesis we extract the existence of functions \[(\alpha,\beta)\mapsto (A(\alpha,\beta),B(\alpha,\beta))\] where $\alpha$ (resp. $\beta$) is the angle that Alice (resp. Bob) could have used to measure the spin projections and $A(\alpha,\beta)$ (resp. $B(\alpha,\beta)$) is the sequence of spin projections thus (conterfactually) measured.

\begin{rmk}
It would be interesting to explore which functional properties $A(\alpha,\beta)$ and $B(\alpha,\beta)$ can have and extract their consequences. For example, it seems reasonable to expect some form of continuity.  Moreover, it could be interesting to explore the existence of such functional properties in the context of a concrete realistic interpretation of QM such as Bohm-de Broglie's theory.  We shall not address these issues in this paper.
\end{rmk}

Among all possible angles $\alpha,\beta$ we distinguish the angles that are experimentally (factually) used and we refer to them as \emph{physical angles} $\theta_a,\theta_b$. In terms of the previous definitions, we thus have $a=A(\theta_a,\theta_b)$ and $b=B(\theta_a,\theta_b)$ or $(a,b)\equiv (a^{\circ\circ}, b^{\circ\circ})$ for simplicity. We use a left (resp. right)  hollow circle in the superscript of $c$ in $\{a,b\}$ to mark that Alice (resp. Bob) has used a physical angle.  We use a filled circle instead to represent the fact that the sequence so marked is a CMR-dependent sequence of values (inferred to exist from CMR along a non-physical angle $\alpha\neq \theta_a$ or $\beta\neq \theta_b$). For instance, besides measurements on both elements of each pair, one may consider the following \emph{alternative situations}:

{\bf{[AS1a]}}  Alice could have chosen another angle of possible measurement $\theta_{a'}$ with Bob keeping the angle $\theta_b$, in which case they would respectively have obtained the sequences $a^{\bullet\circ}\equiv \{a^{\bullet\circ}_i\}$ and  $b^{\bullet\circ}\equiv \{b^{\bullet\circ}_i\}$ (or together the pair of sequences $(a^{\bullet\circ},b^{\bullet\circ})$).

{\bf{[AS1b]}}  Bob could have chosen another angle of possible measurement $\theta_{b'}$, while Alice kept the angle $\theta_a$, in which case they would have obtained the pair of  sequences $(a^{\circ\bullet},b^{\circ\bullet})\,$. 

{\bf{[AS2]}} Alice and Bob could both have chosen the alternative angles $\theta_{a'}$ and $\theta_{b'}$, and thus would 
have obtained the pair of sequences of counterfactual values, $ (a^{\bullet\bullet},b^{\bullet\bullet})$ that is used in  \bonlinecite{dFT}. 

The word \emph{``measured"} when applied to one sequence does not tell us that the corresponding sequence is known: only the co-measured sequences  $(a^{\circ\circ}$ and $b^{\circ\circ})\,$ are  assumed to be known. While the two sequences that make sense together must carry the same super-scripts, we will also consider mixed pairs such as, \emph{e.g.,} $( a^{\bullet \circ}, b^{\circ \bullet})$ (\emph{cf.} (\ref{Ineq1c})).

\begin{rmk} We stress that the sequences which are not measured have a (physical) meaning only when assuming the CMR hypothesis. Moreover, we also stress the fact that, to eventually arrive at a Bell-type contradiction, it is essential to  have such an assumption ensuring the \emph{existence} of the several sequences of $\pm 1$'s which are then bound to satisfy appropriate inequalities (see Boole's lemma below). If one only considers theories yielding probability distributions for outcomes of experiments, the concrete sequences of $\pm 1$'s need not exist and, thus, the inequalities cannot be used to extract a contradiction. For further discussions, see \bonlinecite{Zu}.
\end{rmk}

Recall now that \emph{locality} means that the outputs (either measured or inferred to exist by assuming CMR) do not depend on the choice of vector made by the other observer when the measurements by Alice and Bob are {space-like separated}.  We do not assume locality here as it is known since \bonlinecite{Bell} that assuming QM, CMR, and locality implies a contradiction (indeed, CMR instead of the Predictive Hidden Variables used in \bonlinecite{Bell} works the same). 

\begin{rmk} With the set of hypotheses made so far, there is no reason to expect (as holds true when assuming locality) 
that {$c^{s_1\,t_1} =c^{s_2\, t_2}$ for $c\in \{a,b\}$ and $s_i,t_i\in\{\circ,\bullet\}$,  beyond the case when $s_1=s_2$ and $t_1=t_2$}. In particular, eight different sequences could have been generated,  {namely} $a^{\circ\circ},b^{\circ\circ},a^{\circ\bullet },b^{\circ\bullet},$ $a^{\bullet\circ},b^{\bullet\circ},a^{\bullet\bullet}, b^{\bullet\bullet}$. 
\end{rmk}                       

\medskip
Without further assumption than QM, the twisted Malus law lets us compute $\inner{a^{\circ\circ }}{b^{\circ\circ}}$,  
but when we assume both QM and CMR, more observables values make sense, and it follows {from} Conventions 1 and 2 that this law applies to more pairs, so that we have (see also \bonlinecite{dFT}):

\begin{eqnarray}\label{eq:tmls}
\inner{a^{\circ\circ }}{b^{\circ\circ }}\,=\,-\cos(\theta_{ab}), &  \inner{a^{\circ\bullet }}{b^{\circ\bullet }}\,=\,-\cos(\theta_{ab'}), \nonumber\\ 
\inner{a^{\bullet\circ }}{b^{\bullet\circ }}\,=\, -\cos(\theta_{a'b}), & 
\inner{a^{\bullet\bullet}}{ b^{\bullet\bullet}} \,=\, -\cos(\theta_{a'b'}) \,. 
\end{eqnarray}

\section{Boole's lemma}\label{subsec:boole}
Consider now the inequalities
\[|x_1 + x_2|+ |x_3 - x_4| \leq 2,\] \[|x_1 - x_2|+ |x_3 + x_4| \leq 2.\]
The {points in} ${\bm{R}}^4$ that satisfy both inequalities 
{define} a polytope ${\bm{P}}  \subset I^4 \subset {\bm{R}}^4 $,
where $I=[-1,+1]$. Let us recall \bonlinecite{Boole1862}, \bonlinecite{Pitowsky2001}:

{
\begin{customthm} 
With $u$, $x$, $v$, and $y$ denoting sequences of numbers in $\{-1,1\}$, the inequality 
\begin{equation}\label{Boole} 
|\inner{u}{x} + \inner{v}{x}|+ |\inner{u}{y} - \inner{v}{y}| \leq 2 
\end{equation}
holds, provided that the limits defining the 4 correlations involved in \emph{(\ref{Boole})} exist.
\end{customthm}
}
 
Notice that we could apply Boole's lemma to all correlations associated  to the sequences
$a^{\circ\circ},b^{\circ\circ},a^{\circ\bullet },b^{\circ\bullet},$ $a^{\bullet\circ},b^{\bullet\circ},a^{\bullet\bullet}, b^{\bullet\bullet}$ obtained before. However only four of these correlation are used 
together.  Also, following a tradition that goes back to \bonlinecite{CHSH69}, we only consider correlations of the form {$\inner{a^{s_1\,t_1}}{b^{s_2\,t_2} }$}.  
Since all the sequences that we consider in these correlations have values in $\{-1,+1\}$, it follows from 
Boole's Lemma under standard convergence hypotheses that one has the \emph{CHSH inequalities} \bonlinecite{CHSH69}:
\begin{eqnarray}\label{CHSHeq}
|\inner{a^{\circ\circ}}{b^{\circ\circ}} + \inner{a^{\bullet\circ}}{b^{\circ\circ}}|+ 
|\inner{a^{\circ\circ}}{b^{\circ\bullet}} - \inner{a^{\bullet\circ}}{b^{\circ\bullet}}| \leq 2 \,\,\,\label{Ineq1c}  \,\,\,\,\,\,\,\,  \nonumber \\
|\inner{a^{\circ\circ}}{b^{\circ\circ}} - \inner{a^{\bullet\circ}}{b^{\circ\circ}}|+ 
|\inner{a^{\circ\circ}}{b^{\circ\bullet}} + \inner{a^{\bullet\circ}}{b^{\circ\bullet}}| \leq 2. \,\,\,\,\,\,\,\, \label{Ineq2c}
\end{eqnarray}

With the hypotheses that we have made so far, the correlations $\inner{a^{\bullet\circ}}{b^{\bullet\circ}}$,
$\inner{a^{\circ\circ}}{b^{\circ\bullet}}$, and $\inner{a^{\bullet\circ}}{b^{\circ\bullet}}$ are all unknown to us and three of the known correlations \eqref{eq:tmls} still cannot be used. 

In the next section, we will introduce a principle that allows us to relate the correlations appearing in the inequality \eqref{CHSHeq} {with} those of \eqref{eq:tmls}.  {Geometrically, the CHSH inequalities state that the $4$-tuple of correlations appearing in both of them must lie in the polytope $\bm{P}$ defined above.}

\section{The restricted effect-after-cause principle and its consequences}\label{sec:REACP}

In what follows we are dealing with observables whose measured values are registered so that changes of said values that would happen after measurement could be compared to the original values.

The \emph{restricted effect after cause principle} (\emph{REACP}) states that, \emph {for any Lorentz observer, the registered value of an observable cannot change because of further choices that are made after registration.}

\emph {In particular, the REACP states that, for any Lorentz observer, the registered observable value for a physical choice of angle cannot change because of further choices that are made after that.} 

We shall prove later (section \ref{sec:disc}) that the REACP is a direct consequence of SR.
When assuming the REACP we do not ask from CMR-dependent values the same that is asked from actual values that are generated by measurements because of the following simple observation (in the spirit of the title of a famous short paper by Peres \bonlinecite{Peres1993}): 

- \emph{What happens may have consequences but what could have happened (at the atomic scale at least) has no consequences (on the subsequent actual world).}  

- In particular,\emph{``Only an actual measurement by Alice may have an effect on the sequence measured by Bob, and \emph{vice versa}.''}

%new
\begin{rmk}
It is important to notice that, even in a Bell-type setting in which what is to be measured is picked randomly, the assumption of the REACP only applies 
to the value of the only physically measured quantity and not to the other possible `hidden values' of the underlying realistic theory. 
The fact of being measured, even when picked randomly, distinguishes the registered value from the rest and it is on this registered value that the 
REACP poses any constraints. From this perspective, the REACP is clearly weaker than standard locality, which involves all possible measurements. 
(See also Section \ref{sec:disc}.)
\end{rmk}

\begin{rmk}\label{toto}
The REACP is the EACP of \bonlinecite{Tresser1}, \bonlinecite{Tresser2} restricted to actual measurements.
The definition of the EACP seems excessively convoluted, something important here since 
we shall use a lot the specialization of that principle in the form of the REACP.   One may wish to replace the EACP by some lighter principle, such as, \emph{``a measurement value, be it factual or realism based, can only depend on what happens in the past cone of the measurement or of the more general observable value attribution"}.  
When using such a definition, one can compute the same correlations that one accesses when assuming locality.  
While some would say that locality is nevertheless a stronger statement that such a ``gentle EACP" (because locality tells us that no future cone 
event whatsoever can matter), it appears that, as far as applications to the EPRB setting are concerned, locality and such a ``gentle EACP" 
amount to the same sets of technical assumptions.  In particular, by using the right combination of Lorentz observers, one can give a meaning 
(and the value $\cos(\theta_{a'b'})$ that is implied by assuming Locality) to the correlation between the values that depend on the CMR on 
both Bob's and Alice's sides (the correlation that below we denote $\inner{a^{\bullet\circ}}{b^{\circ\bullet} }$).
\end{rmk}

We require that measurements for physical angles be Lorentz invariant; more precisely, we require that $a^{\circ\circ}, a^{\circ\bullet}, b^{\circ\circ}, b^{\bullet\circ}$ do not depend on the Lorentz observer (we do not ask anything for the other measurements). 

Recall now that in order to use inequalities (\ref{CHSHeq}) we need to compute $\inner{a^{\bullet\circ} }{b^{\circ\circ} }$, $\inner{a^{\circ\circ}}{b^{\circ\bullet}}$ and $\inner{a^{\bullet\circ}}{b^{\circ\bullet} }$, or at least we need to know enough about these correlations.

\begin{claim}
\label{conv1bis} The REACP implies that $b^{\bullet\circ} =b^{\circ\circ}$ {\rm and} $a^{\circ\bullet} =a^{\circ\circ}.$
\end{claim}

\begin{proof}
This claim follows from choosing different Lorentz observers, as we now explain. 

- (i) Consider one Lorentz observer for which the measurements made by Bob happen before the corresponding ones made by Alice. The REACP tells us directly that no change later performed by Alice can affect the measurement done by Bob, meaning that his measurement $b^{\bullet\circ}=B(\theta_{a'},\theta_b)$ only depends on the physical angle $\theta_b$ chosen by him. Namely, we get that $B(\theta_{a'},\theta_b)$ is independent of $\theta_{a'}$ and, since we assume that this measurement is Lorentz invariant, that  $b^{\bullet\circ}=B(\theta_b) =b^{\circ\circ}$.  

- (ii) Considering instead another observer for whom Alice measures before Bob on corresponding particles, we get $a^{\circ\bullet} =a^{\circ\circ}$ \emph{mutatis mutandis}. \end{proof}

\begin{corollary} \label{Coro} 
Assuming QM and CMR it follows from the REACP and equations (\ref{eq:tmls}) that $\inner{a^{\bullet\circ}}{b^{\circ\circ}}=\inner{a^{\bullet\circ}}{b^{\bullet\circ}}=-\cos(\theta_{a'b})$ 
and that $\inner{a^{\circ\circ}}{b^{\circ\bullet}}= \inner{a^{\circ\bullet}}{b^{\circ\bullet}}=-\cos(\theta_{ab'}).$
\end{corollary}

\section{Main claims}
We shall prove some weak regularity properties of the correlation $\inner{a^{\bullet\circ}}{ b^{\circ\bullet}}$ as a function of certain angle quadruplets.  We shall also prove that we get a contradiction if we further make the mild assumption that this function is a bit more regular.

Let us set ${\bm{T}}\equiv [-\pi,\pi]/(-\pi \sim \pi)$ and $I\equiv[-1,1]$. There is a function ${\cal{C}}: {\bm{T}}^4 \to I^4$ which represents 
how the following correlations depend on choices of quadruplets of angles $\overline{\theta} = (\theta_a,\theta_{a'},\theta_b,\theta_{b'})$ 
in ${\bm{T}}^4.$  We have:
\begin{align*}
&{\cal{C}}(\overline{\theta}) = (\inner{a^{\circ\circ}}{ b^{\circ\circ}}, \inner{a^{\bullet\circ}}{b^{\circ\circ}}, \inner{a^{\circ\circ}}{b^{\circ\bullet}}, 
\inner{a^{\bullet\circ}}{b^{\circ\bullet}})\\
&= (-\cos(\theta_{ab}),-\cos(\theta_{a'b}),-\cos(\theta_{ab'}),\inner{a^{\bullet\circ}}{b^{\circ\bullet}}(\overline{\theta})).
\end{align*}
Observing that $(\theta_{ab}+\theta_{a'b'})-(\theta_{ab'}+\theta_{a'b})=0$ we consider the variables 
$\theta_1, \theta_2, \theta_3,$ and $\theta_4$ defined as:
\[ 
\theta_1=\theta_{ab},\,\, \theta_2=\theta_{a'b},\,\, \theta_3=\theta_{ab'},\,\,\theta_4=\theta_1-(\theta_2+\theta_3).
\]
Let us denote by $\bm{\theta}$ the triplet of angle differences $(\theta_{ab}, \theta_{a'b}, \theta_{ab'})$, so that  
$\bm{\theta}= (\theta_1,\theta_2,\theta_3)$.  Also, let:
\[
{\cal{Q}}_+ = \{\bm{\theta} = (\theta_1,\theta_2,\theta_3)\in {\bm{R}}^3 : 0\leq   \theta_i \leq \pi, i=1,2,3\}\,,
\]
and consider the function $F_4 : {\cal{Q}}_+ \to {\bm{R}}$ defined by  $F_4(\bm{\theta})=\inner{a^{\bullet\circ}}{b^{\circ\bullet}}$.

The CHSH inequalities \eqref{Ineq2c} then become
\begin{eqnarray}
 |-\cos(\theta_1) - \cos(\theta_2)|+|-\cos(\theta_3) - F_4(\bm{\theta})|&\leq& 2 \nonumber \\
 \label{eq:ineq_cos} |-\cos(\theta_1) + \cos(\theta_2)|+|-\cos(\theta_3) + F_4(\bm{\theta})|&\leq& 2,
\end{eqnarray}
for the unknown function $F_4$.

\begin{prop}\label{prop:diff_0}
 Assume $F_4$ is a function that satisfies \eqref{eq:ineq_cos}.  Then $F_4$ is differentiable at ${\bm{0}}=(0,0,0)$  and $DF_4({\bm{0}}) = (0,0,0)$.
\end{prop}

\begin{proof}
From (\ref{eq:ineq_cos}) and the definition of $F_4$ it follows that $F_4(\bm{0}) = -1$.  Moreover, for $\bm{\theta}$ close to ${\bm{0}}$,  the CHSH inequalities imply:
$
|F_4(\bm{\theta}) + 1| \leq 3 - \cos\theta_1 - \cos\theta_2 - \cos\theta_3\,.
$
Dividing by the norm of ${\bm{\theta}}$, we get a quotient that clearly goes to zero as $\bm{\theta} \to {\bm{0}}$.
\end{proof}

\begin{rmk}
Following the proof above, it is easy to infer the more general result that $F_4$ is differentiable at any point $\theta \bm{e}_i\in {\cal{Q}}_+$ (with $\bm{e}_i \in \bm{R}^3, i=1,2,3$ the standard basis vectors) with $DF_4|_{\theta \bm{e}_i}= \sin(\theta) \bm{e}_i$. Moreover, one can also show that the second order variation $|F_4(\bm{\theta})+1|/\|\bm{\theta}\|^2$ remains bounded as  $\bm{\theta}\to {\bm 0}$.
\end{rmk}

In terms of the physical setting, we have shown:
\begin{theorem} \label{diff0}
Suppose {\emph{[}}QM+CMR+REACP{\emph{]}}. Then, $\inner{a^{\bullet\circ}}{b^{\circ\bullet}}$ is differentiable at ${\bm{0}}=(0,0,0)$ with zero differential.
\end{theorem}

We now show that, assuming mild further regularity of $F_4$ at ${\bm{0}}\in{\cal{Q}}_+$, a Bell-type contradiction arises. 
For that, we first need a short lemma that follows from QM and CMR.

\begin{lemma}\label{lma:diag}
 Assuming {\emph{[}}QM+CMR{\emph{]}}, then $F_4(\theta,\theta,\theta)=-\cos(\theta)$.
\end{lemma}

\begin{proof}
By direct computation using the definitions, $\theta_i = \theta \ \forall i=1,2,3$ implies that $\theta_{a'}=\theta_a$ and $\theta_{b'}=\theta_b$.  So, in such case, we have $a^{\bullet\circ}=a^{\circ\circ}$ and $b^{\circ\bullet}=b^{\circ\circ}$, so that $\inner{a^{\bullet\circ}}{b^{\circ\bullet}}=\inner{a^{\circ\circ}}{b^{\circ\circ}}=-cos(\theta)$ by the twisted Malus Law.
\end{proof}

\begin{prop}\label{prop:3=1}
Assume that $F_4$ is a function satisfying \eqref{eq:ineq_cos}, that $F_4(\theta,\theta,\theta)=-\cos(\theta)$ and also that $F_4$ is twice differentiable at $\bm{\theta} = {\bm{0}}$. Then, it follows that $3=1$.
\end{prop}

\begin{proof}
We illustrate the main point of the proof by assuming that 
$F_4$ is a degree $2$ polynomial near $\bm{\theta} = {\bm{0}}$, whence
$
F_4{(\bm{\theta})}=-1 + \sum_i c_i \theta_i + \sum_{i\leq j} c_{ij} \theta_i \theta_j\,,
$ 
(the general case follows readily by using the elementary theory of Taylor polynomials.)  We accordingly replace the function $-\cos (x)$ by its order-two approximating polynomial $-1 + \frac{1}{2} x^2$.  By Proposition \ref{prop:diff_0}, all linear coefficients must vanish: $c_i = 0$ for $i=1,2,3$.  Next, it is not hard to check that if $|z| \leq 1$, $|w| \leq 1$ and $| 1 \pm z| + |1 \mp w| \leq 2$ then $z =w$ ({\it cf.\/} \bonlinecite{dFT}).  In this way, we thus obtain that $c_{11} = \frac{1}{2}$ by considering $\bm{\theta} = (x,0,0)$.  Analogously,  we obtain $c_{22} = c_{33} = \frac{1}{2}$. Now, taking $\bm{\theta} = (x,x,x)$ and noting that $F_4$ coincides with the cosine on the diagonal, we have: 
\begin{eqnarray}\label{cont}
\frac{1}{2} + \frac{1}{2} + \frac{1}{2} + \sum_{i < j} c_{ij} = \frac{1}{2}.
\end{eqnarray}
To finish the proof, we show that $c_{ij} = 0$ when $i < j$.   For that purpose, consider the second order terms in the CHSH inequalities, taking $\bm{\theta} = (0, x , k x)$ (with $x$ small and $kx$ also small but fixed); we get that for all $k$:
\[
{2 - \frac{1}{2} x^2 + \left|\frac{1}{2} x^2 + c_{23} k x^2\right| \leq 2 + o(x^2)\ .}  
\]
{Canceling the $2$ from both sides}, dividing by $x^2$ and taking $k$ arbitrarily large, we see that the inequality holds only when $c_{23}=0$.  The equalities $c_{13} = 0 = c_{23}$ are shown in the same way. Therefore, equation (\ref{cont}) implies that $3=1$. 
\end{proof}

\begin{rmk}
 It is interesting to notice that a similar argument can be applied for other angle configurations $\bm{\theta}\neq \bm{0}$ 
which correspond through $\cal{C}$ to vertices of the polytope $\bm{P}$ defined by the inequalities.
\end{rmk}

Our main result now follows by combining Lemma \ref{lma:diag} and Proposition \ref{prop:3=1}:
\begin{theorem}\label{maint}
Suppose {\emph{[}}QM+CMR+REACP{\emph{]}} and that $\inner{a^{\bullet\circ}}{b^{\circ\bullet}}$ is twice differentiable at  $\bm{\theta} = {\bm{0}}$ as a function restricted to  ${\cal{Q}}_+$.  Then, it would follow that $3=1.$
\end{theorem} 

Assuming one extra degree of differentiability for a correlation that we  prove  differentiable would appear \emph{a priori} much weaker than assuming the first degree of differentiability.  Since, as we shall see in \S \ref{sec:disc}, REACP is a consequence of SR, Theorem \ref{maint} tells us that either CMR is false or the correlation $\inner{a^{\bullet\circ}}{b^{\circ\bullet}}$ is twice differentiable. 

The following three remarks are formulated under the assumption that CMR holds true. 

\begin{rmk}
Notice that Proposition \ref{prop:3=1} implies a Bell-type contradiction whenever one can be sure that $\inner{a^{\bullet\circ}}{b^{\circ\circ}}$ and $\inner{a^{\circ\circ}}{b^{\circ\bullet}}$ are given by the twisted Malus law. In this paper, we showed that the REACP leads to this conclusion and it would be interesting to find other such hypotheses.
\end{rmk}

\begin{rmk}
Since Theorem \ref{maint} only requires functional regularity when restricted to ${\cal{Q}}_+$, it also holds when the correlation $\inner{a^{\bullet\circ}}{b^{\circ\bullet}}$ or its derivatives are allowed to jump as some of the $\theta_i$'s ($i=1,2,3$) changes sign; in particular, Theorem \ref{maint} can be applied to correlations that depend on absolute values of corresponding angle differences. 
\end{rmk}

\begin{rmk}
An example of a correlation function that escapes the contradiction of Theorem \ref{maint} is given by $\inner{a^{\bullet\circ}}{b^{\circ\bullet}}$ having the value $-cos(\theta)$ when $\theta_i=\theta \ \forall i=1,2,3$ and $-cos(\theta_1)cos(\theta_2)cos(\theta_3)$ otherwise.  As in this example, a function escaping the contradiction must have 'jumps' when varying the direction of the vector $\bm{\theta}\in \bm{R}^3$.
\end{rmk}

\section{The REACP and Special Relativity}\label{sec:disc}

We first show that the REACP is weaker than SR, something that we could not prove (so far) for the effect after cause principle of \bonlinecite{Tresser1} and \bonlinecite{Tresser2}. 

\begin{claim}\label{REACPWeaker} The REACP is implied by SR, and more specifically by the impossibility of sending superluminal messages.\end{claim}

\begin{proof}
To prove this statement we show that the negation of the REACP allows superluminal signaling.  
To fix the ideas, assume that two slots $-1_i$ and $+1_i$ are associated to the $i$-th event 
in a sequence of events labeled by $i$.  The result of an experiment is then marked by putting a slot 
marker in one of these two slots.  Thus, if the REACP fails, the marker of some slot would move 
because of later manipulations in the time frame of some Lorentz observer.  Even if only some 
sufficient proportion of slots would get removed or displaced, standard error correcting techniques \bonlinecite{ECC} by 
repetition of symbols and intertwining of words (to better correct bunched errors) would allow to send YES or NO messages with high probability backward in time, which is the basic form of superluminal transmission.  This argument also shows that the negation of the REACP is stronger than the negation 
of SR, so that the REACP follows from SR.  This proves Claim \ref{REACPWeaker}. 
\end{proof}

On the other hand, we know (see, \emph{e.g.,} \bonlinecite{Eberhard1978}) that the negation of locality does not imply the possibility {of transmitting} superluminal \emph{messages} (symbolic sequences with information content).  From {that} we deduce the following result:
\begin{corollary}\label{REACPandRegWeaker} The conjunction of the REACP with the regularity hypothesis of Theorem \ref{maint} is strictly weaker than locality when assuming {\emph{[}}QM +CMR{\emph{]}}.
\end{corollary}

\section{Discussion and conclusions}\label{sec:last}
In conclusion,  Theorems \ref{diff0} and \ref{maint}, Claim \ref{REACPWeaker}, and Corollary \ref{REACPandRegWeaker} {taken together} express a non-realism statement that is real progress on the traditional form of Bell's Theorem. 
As Mermin states in \bonlinecite{Mermin}(page 812), locality is a more compelling assumption than non-contextuality. From this point of view, our results should also be compared with the quite interesting Bell-Kochen-Specker type results that have flourished recently.  Indeed, for {those} who believe {in} both QM and (the kinematics part of) SR and are willing to sacrifice {CMR}, the only {small} gap that remains to {achieve} a full replacement of von Neumann's faulty theorem is one degree of differentiability {for} one of the correlations that come {up} in Bell's theory, a correlation that we do prove to be differentiable at 0, in contradiction with what one expects from a CMR-compatible correlation \bonlinecite{Bell3}. 

For those who want to keep {CMR} as part of a description of the Quantum world while preserving (the kinematics part of) SR, our result constrains the kind of 
correlation {functions} that could appear in such theories.  In particular, it would be interesting to know the implications of the present work in Bohm-de Broglie's 
interpretation of QM.  Although traditional formulations of Bohmian mechanics seem to be fundamentally incompatible with Special Relativity, recent work by 
D\"urr \emph{et al}, \bonlinecite{BL}, suggests a ``formulation of Bohmian theories that seem to qualify as fundamentally Lorentz invariant''.  In that case, 
even if the underlying correlations are not fully computable, it would be interesting to infer 
some of their qualitative properties and compare them with the restrictions expressed in Theorems \ref{diff0} and \ref{maint}; 
as well as to deduce the physical meaning of the mandatory `jumps' in the correlation $\inner{a^{\bullet\circ}}{b^{\circ\bullet}}$ 
within that interpretation.

\medskip 
\textbf{Acknowledgments:} 
This work was made possible by support from FAPESP through ``Projeto Tem\'atico {\it Din\^amica  em Baixas Dimens\~oes\/}'', Proc. FAPESP {2011/16265-2}, and also Proc. FAPESP 2012/19995-0. The  support of Palis-Balzan's Prize is also gratefully acknowledged. We thank A. Fine for his comments on a former version of the manuscript.

\end{document}